\documentclass[10pt]{llncs}
\usepackage[utf8]{inputenc}

\usepackage{amsmath}

\usepackage{array}
\usepackage{stmaryrd}
\usepackage{graphicx}
\usepackage{color}
\usepackage{multirow}

\DeclareMathOperator*{\argmin}{argmin}
\newcommand*{\argminl}{\argmin\limits}

\usepackage[normalem]{ulem}

\newcommand{\pkcssharp}{%
  {\settoheight{\dimen0}{PKCS}PKCS\kern-.05em \resizebox{!}{\dimen0}{\raisebox{\depth}{\#}}}}

\title{Some observations on the optimization of a parallel SHAKE function using Sakura}
\author{Kevin Atighehchi}

\institute{Aix Marseille Univ, CNRS, LIF, Marseille, France\\
\email{kevin.atighehchi@univ-amu.fr}}

\begin{document}

\maketitle


\begin{abstract}
Some parallel constructions of a SHAKE hash function using Sakura coding are introduced, whose basic operation is the Keccak's permutation.
For each proposed tree-based algorithm, observations are made on both its parallel running time (depth) and the required number of processors to
reach it. This preliminary work 
makes the assumption
that the tree-level chaining value length is equal to the capacity of the underlying sponge construction, 
as recommended in the Sakura paper. 
\end{abstract}


\section{Introduction}

Historically, a mode of operation for hashing 
is applied to an underlying compression function having 
a fixed input length in order to process messages of arbitrary length.
However, it can also be applied
to a sequential (and variable input length) hash function merely for the addition of efficiency properties.
To take advantage of parallel architectures, Merkle and Damg{\aa}rd \cite{Dam90,Mer80} have proposed 
tree-based hash modes which allow parallel evaluations of the underlying function on several parts of a message.
To conform with the terminology of Bertoni
\emph{et al.} \cite{BDPV14_Sak,BDPV14_Suf}, we use the term \emph{inner function} for the underlying 
function, denoted $f$,
and the term \emph{outer function} for the function resulting from the application of the mode to $f$.

A mode of operation for hashing \cite{BDPV14_Sak} defines how the message is splitted into substrings and how these latter
are operated using the inner function $f$ and other simple operations. It can be seen as a composition method which 
specifies how the inputs for $f$ are formatted
by the concatenation of these substrings with \emph{frame bits} and possibly one or more chaining values.
Throughout the paper, we use the conventions \cite{BDPV14_Sak,BDPV14_Suf} that a node is a
formatted input to $f$ and that a chaining value is an image by~$f$.

Some SHA-3 candidates have proposed tree hash modes, in particular Skein~\cite{FLSWBKCW09}, MD6~\cite{RABCDEKKLRSSSTY08} and Keccak \cite{BDPV13_keccak}.
A tree mode is also proposed for an 
improved version of Blake \cite{ANWW13} and yet others are investigated specifically for efficient software parallel implementations~\cite{Gue14,AB16,BDPV14_Suf}.
There is still a debate \cite{luk13,kel14} about the way of standardizing tree hash modes. 

Bertoni \textit{et al.} \cite{BDPV14_Suf} give sufficient conditions for a tree based hash function to ensure its indifferentiability
from a random oracle. 
They define the Sakura coding \cite{BDPV14_Sak} which ensures these conditions, and allows any hash algorithm using it to be indifferentiable from
a random oracle, automatically. They also propose to use another tree representation, called hop tree, that we will use throughout the paper.

We address the problem of finding an optimized tree structure for parallel hashing when the nodes are formatted according to Sakura coding. 
Using \nobreak{RawSHAKE} \cite{fips202} as inner function and assuming an unbounded number of processing units, the aim is to construct a SHAKE function that 
decreases the number of parallel steps, where the running time of one step corresponds to the one of the Keccak's permutation. 
We consider that the tree-level chaining value length is equal to the capacity of Keccak, as recommended in \cite{BDPV14_Sak}.
Such a work then focuses
on the optimization in depth and width of a hashing circuit (abstracted with the terminology of parallel computing). This has an interest both 
theoritically and for hardware implementations, especially when the messages to hash are of small size.

The remainder of this paper is organized as follows. Section \ref{prel} contains background information
regarding the inner function RawSHAKE, Sakura coding and the hop tree representation. We will see that a tree of hops is mapped to a tree of nodes.
Section \ref{opti} discusses the optimization of hop trees for minimizing both the parallel running time to process the corresponding tree of nodes and
the number of involved processors.


\section{Preliminaries}\label{prel}

In this section, we assume that the Keccak hash function family \cite{fips202} has 
a state size of $1600$ bits. The rate $r$ of Keccak is then equal to $1600-c$,
where $c$ is its capacity. The Keccak family is simply
denoted $\textrm{Keccak}[c]$. 

Keccak is an algorithm which consists of two phases: the absorbing phase and the squeezing phase \cite{fips202}. During these phases, 
a permutation is iterated on the hash state a number of times which depends on the message size and the digest size. This numer determines
the efficiency of the algorithm.
Throughout the paper, we will assume that the basic operation is one evaluation of the underlying permutation. One unit of time will then refer to
one evaluation of this permutation.

\subsection{RawSHAKE256}

According to Bertoni \emph{et al.} \cite{BDPV14_Sak} and the FIPS 202 standard \cite{fips202}, the
intermediate function RawSHAKE128/256 should be used as inner function $f$ for tree hashing using Sakura. We recall that  
$\textrm{RawSHAKE256}(M) = \textrm{Keccak}[c = 512](M\|11)$, where 11 is the domain separation
suffix of the extendable-output functions. Let us suppose that $f(M)=\textrm{RawSHAKE256}(M)$. 
If $l$ is the bit-length of $M$, the number of iterations of the permutation during the \emph{absorption} phase
is $\left\lceil \frac{l+4}{r} \right\rceil$. The constant $4$ refers to both the \emph{multi-rate} padding of Keccak 
and the domain seperation suffix.
If we denote $d$ the chosen bit-length for the digest of $f$, 
then the total running time of $f$ is 
$$\left\lceil \frac{l+4}{r} \right\rceil + \left\lfloor \frac{d}{r} \right\rfloor.$$

\subsection{Sakura coding}

After having specified a set of conditions
for a tree hash mode to be indifferentiable from a random oracle \cite{BDPV14_Suf}, Bertoni \emph{et al.} proposed 
a tree hash coding meeting these conditions, called Sakura \cite{BDPV14_Sak}. 
This tree coding, whose syntax is specified using the Augmented Backus-Naur Form (ABNF), 
ensures that any tree hash mode compliant with it is sound, \emph{i.e.} that no weaknesses is introduced 
when using the inner function. 

Since the model of the tree using \emph{nodes} is too much general, certain trees of \emph{nodes} cannot be easily represented. 
In Sakura, the notion of \emph{hops} is used for the representation of trees, and a node ($f$-input) contain one or several hops. 
It then allows the encoding of a tree of hops into a tree of nodes. There exist two types of hops: \emph{message hops} containing only message bits
and \emph{chaining hops} containing only chaining values. The notion of \emph{kangaroo hopping} makes it possible to encode several hops into a node,
with the particularity that the first one is a message hop and the followings are chaining hops. Each chaining value that results from kangaroo hopping is called
a kangaroo hop.


\begin{figure}
$\left< final\ node \right> ::=\ \left< node \right>\ '1'$\\
$\left< inner\ node \right> ::= \left< node \right>\ \left<padSimple \right>\ '0'$\\
$\left< node \right> ::= \left< message\ hop \right>\ |\ \left< chaining\ hop \right>\ |\ \left< kangaroo\ hopping \right>$\\
$\left< kangaroo\ hopping \right> ::= \left< node \right> \left< padSimple \right> \left< chaining\ hop \right>$\\
$\left< message\ hop \right> ::= \left< message\ bit\ string \right>\ '1'$\\
$\left< message\ bit string \right> ::=\ ''\ |\ \left< message\ bit\ string \right>\ \textrm{MESSAGE\_BIT}$\\
$\left< chaining\ hop \right> ::= \textrm{nrCVs} \left< CV \right>\ \left< coded\ nrCVs \right>\ \left< interleaving\ block\ size \right>\ '0'$\\
$\left< CV \right> ::= n\textrm{CHAINING\_BIT}$\\
...\\
$\left< padSimple \right> ::=\ '1'\ |\ \left< padSimple \right>\ '0'$
\caption{Some production rules of Sakura tree coding}
\label{ABNF_rules}
\end{figure}

The hops form a tree whose root is the final hop. A tree of hops uniquely determines a tree of nodes. In a hop tree, each hop (except the final hop)
has an outgoing edge. A message hop has no incoming edges. The degree of a chaining hop corresponds to the number of its incoming edges, and
the hops at the other end of these edges are its child hops. The indexing of a hop
is defined in a recursive way: starting from the final hop of index the empty sequence (denoted~$*$), the 
$i$-th child of a hop with index $\alpha$ has index $\alpha \| i - 1$.

In a tree of nodes ($f$-inputs), there are also two types of nodes: \emph{inner nodes} and \emph{final nodes}. 
The result of $f$ applied to the final node is the output of the outer (constructed) hash function. The last hop in the final node (of 
a tree of nodes) corresponds to the final hop (in the tree of hops).
Figure \ref{ABNF_rules} gives a 
sample of the Sakura rules to format the nodes 
with hops and \emph{frame bits}. We remark that kangaroo hopping can be applied recursively to produce multiple chaining hops.
The $\left< coded\ nrCVs \right>$ production rule encodes the number of chaining values
and requires $\lfloor \log_{256}(n_{cv}) \rfloor + 2$ bytes. 
The \emph{interleaving block size} is coded using two bytes and is not presented in detail since it is of no interest in the present paper.
We refer the reader to the Sakura paper~\cite{BDPV14_Sak} for more information.

The illustrations used in this paper correspond to trees of hops following roughly the same conventions of Bertoni \emph{et al.}~\cite{BDPV14_Sak}.
Message hops have sharp corners, chaining hops have rounded corners. The final hop has a grey fill, the others a white fill. 
An edge between child and parent has an arrow and enters the parent from below if the chaining value obtained by applying $f$ 
to the child hop is in the parent hop. In the case of kangaroo hopping, it has a diamond and enters the parent hop from the left . 
Hops on the same horizontal line are in the same node.

\subsection{Costs of nodes}

Let $l$ the size of the message in a \emph{message hop}, $d$ the size of a chaining value, and $n_{cv}$ the number of chaining values in a \emph{chaining hop}. 
Assuming that no extra bits 0 are used, the appending/prepending of frame bits using Sakura coding yields the following node sizes:

\noindent
Inner node with only message hop: $l+3$.

\noindent
Inner node with chaining hop: $n_{cv}d + (\lfloor \log_{256}(n_{cv}) \rfloor + 1)8+27$.

\noindent
Inner node with kangaroo hop (message hop + chaining hop): $l+2+n_{cv}d+(\lfloor \log_{256}(n_{cv}) \rfloor + 1)8+27$.

\noindent
Final node with only message hop: $l+2$.

\noindent
Final node with chaining hop: $n_{cv}d + (\lfloor \log_{256}(n_{cv}) \rfloor + 1)8+26$.

\noindent
Final node with kangaroo hop (message hop + chaining hop): $l+2+n_{cv}d+(\lfloor \log_{256}(n_{cv}) \rfloor + 1)8+26$.~\\

\paragraph{\textbf{Some redefinitions.}} Note that the bits due to the \emph{multi-rate} padding of Keccak and the domain separation suffix of RawSHAKE have not been counted. Indeed, the definition of a node
by Sakura does not take them into account. In the rest of the paper,
we consider that these bits are an integral part of a node, and, as a consequence, 
we have to consider that $f$ (\emph{i.e.} RawSHAKE) is devoid of both the multi-rate padding bits and the domain separation suffix (\emph{i.e.} $10^*111$). 
Since our purpose is to increase the number of message bits processed by the underlying permutation of Keccak, it is normal to
reduce the number of frame bits as much as possible (\emph{i.e.}, the extra bits 0 are removed). For the same reasons, by using our redefinitions, we 
can add the constraint that $f$ requires an input of length $1088k$ bits, where $k$ is a strictly positive integer.

\section{A specification of SHAKE256 via an optimized circuit}\label{opti}

To construct an optimized circuit for SHAKE256, we first proceed in two phases:
\begin{enumerate}
 \item We seek a hop subtree having a single kangaroo hop that maximizes the number of message bits encapsulated, and 
 whose the mapped tree of nodes can be processed in parallel in a fixed and small amount of time (for instance 2 units of time).
 In this phase, we are interested in the parallel running times that are optimal for the number of message bits processed. Thus, considering
 the criteria of a single kangaroo hop, we will see that there exists a parallel running time threshold above which these hop subtrees 
 are no longer optimal.
 \item The resulting root hops are connected via another tree of (kangaroo) hops. The inner function being RawSHAKE256, we will see in the subsection 
 \ref{multiple_chaining_hops_per_node} that a ternary hop tree is the most appropriate.
\end{enumerate}
We shall see in the subsection \ref{one_chaining_hop_per_node} that the resulting tree of nodes (mapped from the hops) 
is processed in parallel via RawSHAKE with a near optimal running time. Thereafter,
we will see how this tree can be further optimized.

\subsection{Maximizing the bitrate during a small and fixed amount of time}

With the assumption that no extra bits 0 are used for padding, we seek hop trees of height 1 which use a single chaining hop and 
maximize the bitrate during a small and fixed amount of time. In order to increase the bitrate, we necessarily need to use kangaroo hopping.
This single chaining hop is then called a kangaroo hop.
We have to consider whether or not the kangaroo hop is the final hop, 
or, in other words, whether or not the kangaroo hop is in the final node.
We first look at the amount of data that can be processed in 2 units of time:

\subsubsection{Case of an inner node.} If it is processed via $f$ in one unit of time, this node 
contains at most 1081 message bits using a message hop (without kangaroo hopping). 
If it is processed in 2 units of time, this node makes use of kangaroo hopping in the following way:
it contains a message hop with at most 1111 message bits, followed by a chaining hop 
having 2 chaining values of size $c$ bits, with $c=512$. Consequently, if 3 processors are used to compute in parallel 3 nodes, 
it is possible to process 3273 bits of the message (2162+1111) in 2 units of time.

\subsubsection{Case of a final node.}
If it is processed via $f$ in 2 units of time, this node 
contains a message hop having at most 1112 message bits, followed (using kangaroo hopping) by a chaining hop that 
contains 2 chaining values of size $c$ bits, with $c=512$. Consequently, if 3 processors are used to compute in parallel 3 nodes, 
it is possible to process a message of 3274 bits (2162+1112)
in 2 units of time. The corresponding hop tree is depicted Figure \ref{small_hop_tree_2_units}.

\begin{figure}[!h]
 \centering
 \scalebox{0.4}{
 \input{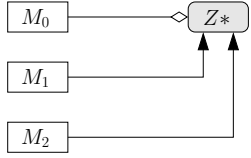}}
 \caption{\textbf{Small hop tree with a single kangaroo hop.} The message hop $M_0$ contains 1112 message bits, while $M_1$ and $M_2$ each contain
 1081 message bits. The chaining (kangaroo) hop is the final hop of the tree.}
 \label{small_hop_tree_2_units}
\end{figure}

Note that this small hop tree will be used in the following section to construct a hop tree covering a message of arbitrary length.
In the mean time, it is interesting to have an idea of the amount of data that can be processed in 3 units of time and using a single kangaroo hop.
The following example is given:

\subsubsection{Case of an inner node.} 
If it is processed via $f$ in 2 units of time, this node 
contains at most 2169 message bits using a message hop (without kangaroo hopping). 
If it is processed in 3 units of time, this node makes use of kangaroo hopping in the following way:
it contains a message hop with at most 2199 message bits, followed by a chaining hop 
having 2 chaining values of size $c$ bits, with $c=512$. Consequently, if 3 processors are used to compute in parallel 3 nodes, 
it is possible to process 6537 bits of the message (2199+4338) in 3 units of time.

\subsubsection{Case of a final node.}
If it is processed via $f$ in 3 units of time, this node 
contains a message hop having at most 2200 message bits, followed (using kangaroo hopping) by a chaining hop that 
contains 2 chaining values of size $c$ bits, with $c=512$. Consequently, if 3 processors are used to compute in parallel 3 nodes, 
it is possible to process a message of 6538 bits (2200+4338)
in 3 units of time.~\\

The features of the possible hop subtrees using a single chaining hop (kangaroo hop) are given in Table \ref{small_hop_trees_charac}. 
Only hop subtrees satisfying an optimal running time for the amount of data processed are considered in this table.
In particular, it will become clear in the next section that those that can be processed by $f$ in more than 4 units of time 
are suboptimal (in the sense of the parallel running time). If less than 2169 bits have to be processed, a single message hop (called model 0) should be used. 
If a message of at most 2170 bits have to be processed, a single final message hop should be used\footnote{A ``final'' hop can process one more bit.}.
The resulting node is then processed in at most 2 units of time using a single processor.~\\

\begin{table}[!h]
\centering
\resizebox{11cm}{!}{
 \begin{tabular}{||>{\centering\arraybackslash}p{1.1cm}||>{\centering\arraybackslash}p{2cm}|>{\centering\arraybackslash}p{1.6cm}|>{\centering\arraybackslash}p{1.2cm}|>{\arraybackslash}p{5.5cm}||} 
   \hline
    Hop subtree (model)& Number of message bits encapsulated & Number of processors (nodes) & Parallel running time & Short description (distribution of the message bits among child hops)\\
    \hline
    \hline
    1  & 2704 & 2 & 2 &  1st child of the kangaroo hop: 1623 bits\newline 2nd child hop: 1081 bits\\
    \hline
    2 & 3273 & 3 & 2 & 1st child of the kangaroo hop: 1111 bits\newline 2nd and 3rd child hop: 1081 bits \\
    \hline
    3 & 4880 & 2 & 3 & 1st child of the kangaroo hop: 2711 bits\newline 2nd child hop: 2169 bits\\
    \hline
    4 & 6537 & 3 & 3 & 1st child of the kangaroo hop: 2199 bits\newline 2nd and 3rd child hop: 2169 bits\\
    \hline
    5 & 7106 & 4 & 3 & 1st child of the kangaroo hop: 1687 bits\newline 2nd child hop: 1081 bits\newline 3rd and 4th child hop: 2169 bits\\
    \hline
    6 & 7675 & 5 & 3 & 1st child of the kangaroo hop: 1175 bits\newline 2nd and 3rd child hop: 1081 bits\newline 4th and 5th child hop: 2169 bits\\
    \hline
    7 & 11458 & 4 & 4 & 1st child of the kangaroo hop: 2775 bits\newline 2nd child hop: 2169 bits\newline 3rd and 4th child hop: 3257 bits\\
    \hline
    8 & 13115 & 5 & 4 & 1st child of the kangaroo hop: 2263 bits\newline 2nd and 3rd child hop: 2169 bits\newline 4th and 5th child hop: 3257 bits\\
    \hline
    9 & 13684 & 6 & 4 & 1st child of the kangaroo hop: 1751 bits\newline 2nd child hop: 1081 bits\newline 3rd and 4th child hop: 2169 bits\newline 5th child hop: 3257 bits\\
    \hline
    10 & 14253 & 7 & 4 & 1st child of the kangaroo hop: 1239 bits\newline 2nd  and 3rd child hop: 1081 bits\newline 4th and 5th child hop: 2169 bits\newline 6th and 7th child hop: 3257 bits\\
    \hline
\end{tabular}
}
\vspace{0.2cm}\\
\caption{Characteristics of small and optimal hop subtrees using a single chaining (kangaroo)~hop. Note that the same hop trees having the kangaroo hop as final hop can process
one more message bit, while conserving the other characteristics.}
\label{small_hop_trees_charac}
\end{table}

We 
remark in Table \ref{small_hop_trees_charac} that the hop subtrees of running time 3 or 4 could be replaced by a hop subtree 
obtained by the composition of a certain number of times of the second hop subtree. The running time of the resulting hop subtree remains the same.
For instance, a substitute for the sixth hop subtree (whose the mapped tree of nodes is depicted in Figure \ref{single_kangaroo_hop_model_6}) 
is the following: we divide the message of 7675 bits into parts of 3273 bits, the 
last part being of size 1129. We build a hop subtree for the first part (according to the second model of the table). We do the same for 
the second part and we denote $CV1$ the chaining value obtained when applying $f$ on its longest node. The last part is encapsulated in a single message hop, and
we denote $CV2$ the chaining value obtained when applying $f$ on the corresponding node. We then add a new kangaroo hop in the longest node of the first subtree 
that contains both $CV1$ and $CV2$. Since this kangaroo hop fits completely into the rate width (1088 bits), it is processed in one unit of time. 
If it is possible to obtain an optimal running time by a composition of subtrees in this way, we can ask why we consider the other models. 
This is because they reduce the number of involved processors. For instance, in the substitute that we propose, we need 6 processors, against 5 for the original one.

\begin{figure}[!h]
 \centering
 \scalebox{0.65}{
 \input{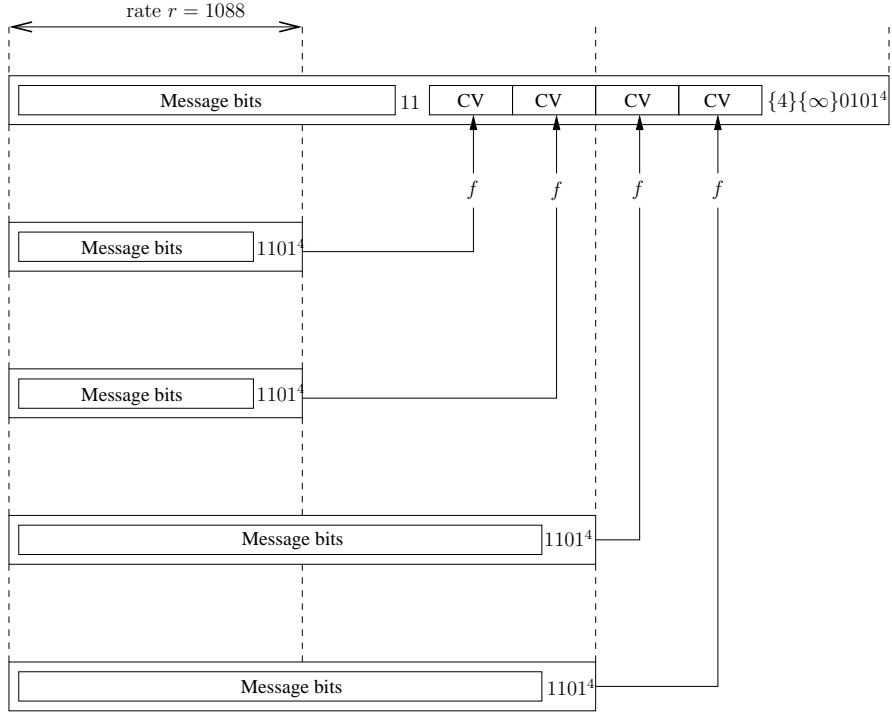}}
 \caption{\textbf{Subtree of nodes mapped from Model 6.} We recall that the representation of nodes is slightly changed. Nodes are expanded 
 with both the multi-rate padding bits of Keccak and the 
 domain separation suffix bits of RawSHAKE256 (these bits, devoid of extra bits $0$, are denoted $1^4$). Each node is processed by the redefined inner function $f$. 
 The chaining values are of length 512 bits and their corresponding boxes are not drawn to scale. In accordance with Sakura, the notation $\{4\}$ represents 
 the two bytes that encode the number of chaining values, and $\{\infty\}$ represents the two bytes that encode the absence of interleaving.}
 \label{single_kangaroo_hop_model_6}
\end{figure}

\paragraph{Recipe for maximizing the number of message bits using a single kangaroo hop, by considering that this kangaroo hop is inside a node of length $1088k$, 
where $k$ is a strictly positive integer.} We count the number of chaining values that can take place inside the last $(k-1)1088$ bits of this node (not forgetting 
the last frame bits $\{n_{cv}\}\{\infty\}0101^4$).
This number is denoted $n_{cv}$. Looking through these $n_{cv}$ chaining values from left-to-right inside this node, we should have the following:
all the chaining values that have bits in the $i$\emph{-th} chunk of 1088 bits should correspond to $f$-image of nodes of size $(i-1)1088$ bits. These nodes
of size $(i-1)1088$ are mapped from message hops, and the number of message bits that can be encapsulated in such a hop is $(i-1)1088-7$. The unoccupied space at the 
begginning of the longest node (\emph{i.e.} at the left side of the first chaining value) serves for the first message hop and the frame bits $11$.

\subsection{Hop trees allowing multiple chaining hops per node}\label{multiple_chaining_hops_per_node}

The first proposed solution, in order to process a message of arbitrary length, is to use the second hop subtree of Table \ref{small_hop_trees_charac}
as many times as needed. 
Note that a kangaroo hop having 3 child hops can be processed via $f$ in one unit of time. Indeed, the chaining values and mandated frame bits it contains 
can fit into 1088 bits. 
We recall that its first child is a message hop or a chaining hop, while the last 
two others are chaining hops.
We thus compose this subtree according to a ternary hop tree\footnote{A hop tree whose each hop (other than the leaf hops) is a kangaroo hop having 
3 child hops.} (see the example depicted Figure \ref{ternaly_hop_tree_1}): the message is divided in parts of 3273 bits, with 
the latter which may be smaller. Each part is encapsulated in a hop subtree using the model 2 of Table \ref{small_hop_trees_charac}, except the last part which,
depending on its size, can use a 
model $j \leq 2$. Thus, we have constructed a certain number of hop subtrees whose number is denoted $p$. A ternary hop tree of height 
$\lceil \log_3 p \rceil$ is then constructed on top of the resulting root hops.
Apart from the first kangaroo hop of a node, each new kangaroo hop has to be aligned according to the rate of Keccak[512]. This ensures that it is entirely processed 
in one evaluation of the underlying permutation\footnote{Padding bits are used to fill the rate of 1088 bits so that each new kangaroo hop inside a node is processed
by a distinct call to the permutation.}.

\begin{theorem}\label{theor1}
Let $n$ be the bit-size of a message. There exists a hop tree for this message
such that all the nodes (mapped from the corresponding hops) can be processed in parallel in $\left\lceil \log_3(\frac{n}{3273}) \right\rceil + 2$ units of time,
using at most $3 \left\lceil \frac{n}{3273} \right\rceil$ processors.
\end{theorem}

\begin{proof}
 Using the model 2 of hop subtree, we can process 3273 bits in 2 units of time. Moreover, using a new kangaroo hop, it is clear that we can process 
 three times as many data in one more unit of time. Thus, we seek the lowest $j$ such that $3^j \times 3273 \geq n$. The final node then contains, in addition to the first kangaroo hop
 contained in a hop subtree of height 1 (model 2), $\lceil \log(\frac{n}{3273}) \rceil$ kangaroo hops. Consequently, the computation of the digest by applying $f$ on 
 the final node requires $\lceil \log(\frac{n}{3273}) \rceil+2$ evaluations of the permutation. The number of involved processors corresponds to the number of nodes,
 \emph{i.e.} at most $3 \left\lceil \frac{n}{3273} \right\rceil$.
\end{proof}

\begin{figure}[!h]
 \centering
 \scalebox{0.28}{
 \input{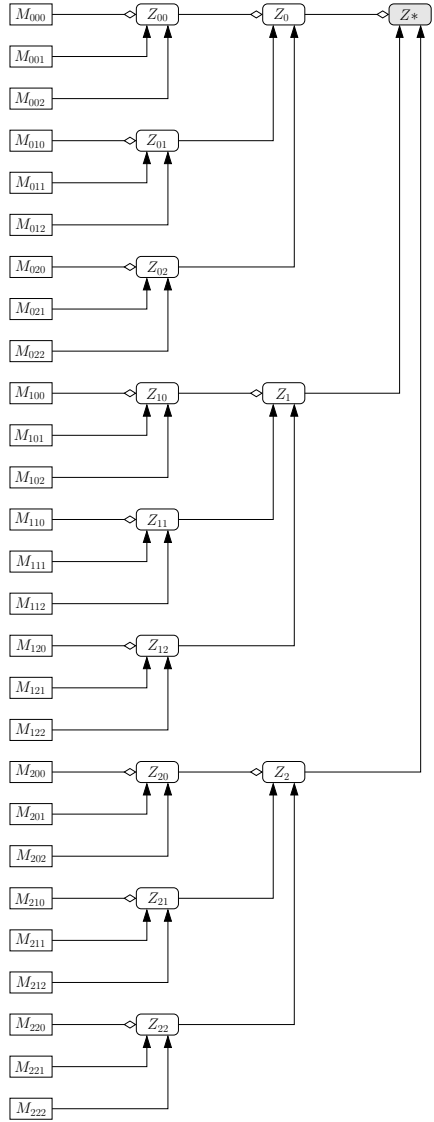}}
 \caption{\textbf{Ternary hop tree for a message of 29457 bits.} The message is divided in parts of 3273 bits and each part is encapsulated using the model 2 of hop subtree. The 9
 corresponding root hops are the leaf hops of a perfect ternary tree of hops. In this example, it appears that 
 the entire hop tree is a perfect ternary hop tree.}
 \label{ternaly_hop_tree_1}
\end{figure}

\newpage

\noindent
For accessing the information of Table \ref{small_hop_trees_charac}, let us denote the following functions:
\begin{itemize}
 \item $N_{mb}(i)$ returns the number of message bits encapsulated in the tree template corresponding to the model $i$ of hop subtree;
 \item $N_{p}(i)$ returns the number of processors (or nodes) required by the model $i$;
 \item $T(i)$ returns the parallel running time to process the instantiated subtree of nodes corresponding to the model $i$ of hop subtree;
\end{itemize}

Let $n$ the bit-size of a message such that $n \geq 3275$. To decrease the number of processors while satisfying a parallel running time of 
$t:=\left\lceil \log_3(\frac{n}{3273}) \right\rceil + 2$, we propose the
following algorithm that returns an appropriate hop tree:
\begin{enumerate}
 \item We denote $S$ the set of indices $i$ that are such that $\left\lceil \log_3\left(\frac{n}{N_{mb}(i)}\right) \right\rceil + T(i)$ is equal to $t$.
 \item We compute the index $i^*:=\argminl_{i \in S} \left\{\left\lceil \frac{n}{N_{mb}(i)} \right\rceil N_{p}(i) \right\}$. \emph{The model $i^*$ of 
 hop subtree is a good
 candidate for reducing the number of processors}.
 \item We divide the message in parts of size $N_{mb}(i^*)$ bits (except the latter which may be of smaller size). Each part is encapsulated using the $i^*$-th hop subtree.
 The $\left\lceil\frac{n}{N_{mb}(i^*)}\right\rceil$ root hops are the leaf hops of a ternary tree. This ternary hop tree is such that each kangaroo hop has 2 chaining values
 and fits entirely into the rate width.
\end{enumerate}

\subsection{Hop trees with at most one chaining hop per node}\label{one_chaining_hop_per_node}

The padding bits used to force the alignment of chaining hops to 1088-bit block boundaries can be considered as a waste. These bits could be used to process some
small parts of the message in order to help decreasing the parallel running time and/or the amount of involved resources.
In this section, we explain ho to get close to optimality by compacting all the chaining values, \emph{i.e.} by using at most one chaining hop per node.
Removing the extra bits 0 for padding complicates the problem since
ensuring the ``happens before'' relationship between the production of chaining values and their use is no longer obvious.~\\

If we look at the hop tree constructed to prove Theorem \ref{theor1} and considering the assumptions made about the padding bits, 
the resulting tree of nodes has the following characteristics:
\begin{itemize}
 \item the number of hop subtrees of height 1 is of the form $3^j$ with an integer $j\geq0$, and since their root hops have 3 leaf hops, 
 the total number of corresponding nodes is of the form $3^{j+1}$.
 \item The longest node is of length $1088(j+1)$ where $j$ is the number of kangaroo hops and $j-1=\left\lceil \log_3(\frac{n}{3273}) \right\rceil$.
 \item $2$ nodes are of length $1088j$.
 \item $2 \times 3$ nodes are of length $1088(j-1)$.
 \item $2 \times 3^2$ nodes are of length $1088(j-2)$.
 \item ...
 \item $2 \times 3^{j-2}$ nodes are of length $1088 \cdot 2$.
 \item $2^2 \times 3^{j-2}$ nodes are of length $1088$.
\end{itemize}

To improve the efficiency while ensuring this ``happens before'' relationship, a first stragegy is to modify the tree of nodes that results 
from the hop tree of the previous subsection, while satisfying the following guidelines:
\begin{enumerate}
 \item the lengths of nodes remain unchanged, \emph{i.e.} each node has a length that is a multiple of 1088 bits and the proportion of nodes of 
 length $1088k$ for a given $k$ remains unchanged.
 \item All the chaining hops of a node are transformed in a single chaining hop. This chaining hop is compacted at the end of the node, \emph{i.e.} no extra padding 
 bits 0 are used at the end. Thus, there are no wasted bits processed by the Keccak permutation. An example of such a tree is depicted Figure \ref{ternaly_hop_tree_2}.
 \item For the sake of simplicity, the number of chaining values (rule $\left< coded\ nrCVs \right>$) is coded using only 2 bytes, \emph{i.e.} we assume
 that a chaining hop cannot have more than 255 chaining values. This is a reasonable assumption in our context.
 \item As previously, no extra bits 0 are used for padding between a message hop and a chaining hop.
\end{enumerate}

\begin{theorem}\label{theor2}
Let $n$ be the bit-size of a message. There exists a tree of nodes encapsulating this message, with at most one chaining hop per node,
and which can be processed in a parallel time of at most $\left\lceil \log_3(\frac{n+31}{3305}) \right\rceil + 2$,
using at most $3 \left\lceil \frac{n+31}{3305} \right\rceil$ processors.
\end{theorem}

\begin{proof}
Following the guidelines above, we have to count the number of message bits encapsulated in the subtrees of height 1: one of them can encapsulate
$2163+1088(j+1) -1024j -41$ message bits. Two of them can each encapsulate $2162+1088j-1024(j-1)-41$ message bits. Six of them can each encapsulate
$2162+1088(j-1) -1024(j-2)-41$ message bits. And so on... $2 \times 3^{j-2}$ of them can each encapsulate $2162 + 2\cdot 1088 - 1024 -41$ message bits.
We want to minimize $j-1$ subject to the constraint that the sum of all these quantities is greater than or equal to $n$. In other words, we seek 
to minimize $j-1$ such that:
$$3^{j-1}3209 + 64j + 128(j\sum_{k=0}^{j-2}3^k - \sum_ {k=0}^{j-2}3^k - \sum_{k=0}^{j-2}k3^k) \geq n-1.$$
We remark that $\sum_{k=0}^{j-2}k3^k=\frac{3^{j-1}(2j-5)+3}{4}$. Consequently, this sum can be simplified to $3^{j-1}3305-32$, yielding the expected result.
\end{proof}

\paragraph{\textbf{Relaxation of the first guideline.}} By respecting the first guideline above, the processors that are responsible for computing nodes having only
a message hop can sit idle for a certain period of time while their result is unused. Let us suppose that such a message hop is the second (or third) child of 
a chaining hop that is contained in a node of length $1088(j+1)$. If we follow the first guideline, the first child hop contains $1088(j+1) - 1024j -41$ message bits,
while the second (or third) child contains $1081$ message bits. It appears that the processor which processes this latter sit idle during 
$\left\lfloor \frac{64j}{1088} \right\rfloor$ units of time. Thus, this message hop
could contain $\left\lfloor \frac{64j}{1088} \right\rfloor1088$ more message bits. In fact, $2 \cdot 3^{j-2}$ nodes can possibly contain more message bits:
two children of the final node can each contain $\left\lfloor 64j/1088 \right\rfloor1088$ more message bits. $2^2$ nodes can each contain 
$\left\lfloor 64(j-1)/1088 \right\rfloor1088$ more message bits. $2^2\cdot 3$ nodes can each contain 
$\left\lfloor 64(j-2)/1088 \right\rfloor1088$ more message bits. And so on... $2^2\cdot 3^{j-3}$ nodes can each contain 
$\left\lfloor \frac{64 \cdot 2}{1088} \right\rfloor1088$ more message bits.
The objective is then to minimize $j-1$ subject to the following constraint:
$$3^{j-1}3305 - 32 +  2\left\lfloor \frac{64j}{1088} \right\rfloor1088 + 2^2 \sum_{k=0}^{j-3} 3^k\left\lfloor \frac{64(j-(k+1))}{1088} \right\rfloor1088  \geq n-1.$$

\paragraph{Remark.} The aforementioned strategy can be applied when the hop subtrees of height 1 follow a model $k \neq 2$. This leads to 
different results and then we have to choose the one which optimizes both the running time and the number of involved processors.

\begin{figure}[!h]
 \centering
 \scalebox{0.28}{
 \input{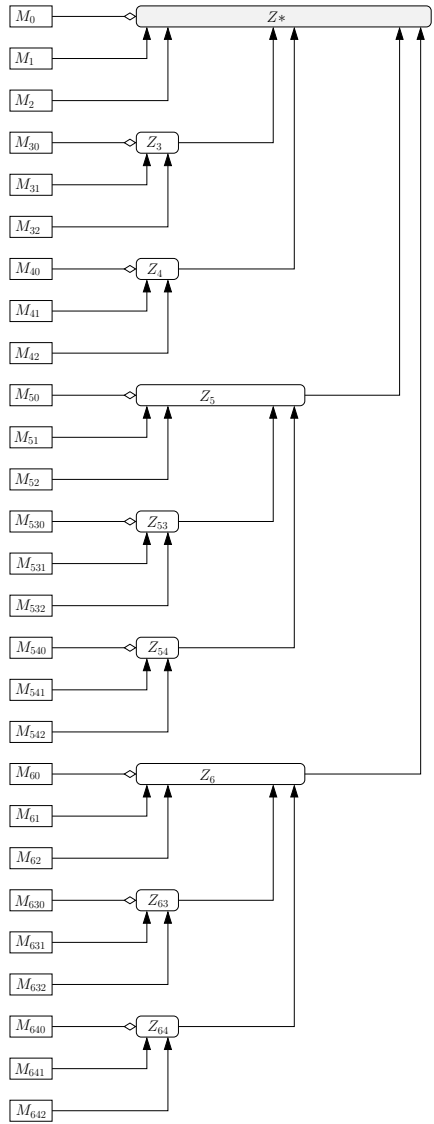}}
 \caption{\textbf{Ternary hop tree with at most one chaining hop per node, encapsulating a message of 29457 bits.} 
 }
 \label{ternaly_hop_tree_2}
\end{figure}

\section{Further remarks}

We have seen how to optimize a tree of nodes using Sakura by taking the example of SHAKE256. 
In order to construct a parallel SHAKE128 function, we need to use RawSHAKE128 as inner function. In this case, the capacity is smaller (256 bits) and the rate 
larger (1344 bits). 
If we consider chaining values to be equal to the capacity, 
then 5 chaining values can fit into the rate width.
All our results need to be recomputed in the case of SHAKE128.


\bibliographystyle{plain}
\bibliography{trees} 

\end{document}